\documentclass[10pt,conference]{IEEEtran}
\IEEEoverridecommandlockouts
\usepackage{graphicx}
\usepackage{subfigure}
\usepackage{epstopdf}
\usepackage{amsthm}
\usepackage{amssymb}
\usepackage{amsmath}
\usepackage{mathtools, nccmath}
\usepackage{cite}
\usepackage{cleveref}
\usepackage{multirow}
\usepackage[letterpaper, left=0.625in, right=0.625in, bottom=1.03in, top=0.7in]{geometry}
\newtheorem{theorem}{Theorem}
\newtheorem{lemma}[theorem]{Lemma}

\begin{document}
\title{Performance of RIS-empowered NOMA-based D2D Communication under Nakagami-$m$ Fading} 
\author{\IEEEauthorblockN{Mohd Hamza Naim Shaikh,
$^{\circ}$Sultangali Arzykulov, 
$^{\circ}$Abdulkadir Celik, 
$^{\circ}$Ahmed M. Eltawil,
and G. Nauryzbayev} 
\IEEEauthorblockA{School of Engineering and Digital Sciences, Nazarbayev University, Nur-Sultan City, 010000, Kazakhstan\\
$^{\circ}$CEMSE Division,
King Abdullah University of Science and Technology, Thuwal, 23955, Saudi Arabia
\\
Emails: \{hamza.shaikh, galymzhan.nauryzbayev\}@nu.edu.kz,\\
$^{\circ}$\{sultangali.arzykulov, abdulkadir.celik, ahmed.eltawil\}@kaust.edu.sa
}
\vspace{-0.7cm}} 
\maketitle

\begin{abstract}
Reconfigurable intelligent surfaces (RISs) have sparked a renewed interest in the research community envisioning future wireless communication networks. In this study, we analyzed the performance of RIS-enabled non-orthogonal multiple access (NOMA) based device-to-device (D2D) wireless communication system, where the RIS is partitioned to serve a pair of D2D users. Specifically, closed-form expressions are derived for the upper and lower limits of spectral efficiency (SE) and energy efficiency (EE). In addition, the performance of the proposed NOMA-based system is also compared with its orthogonal counterpart. Extensive simulation is done to corroborate the analytical findings. The results demonstrate that RIS highly enhances the performance of a NOMA-based D2D network.

\end{abstract}

\section{Introduction}
Reconfigurable intelligent surfaces (RISs) have recently been emerged as a revolutionary technique to realize the smart and programmable wireless environment for the next generation 6G systems \cite{basar2019wireless}. 
Inherently, RIS consists of a planar array of large number of passive reflecting elements (REs). These REs can reflect the incident signal so that the reflected signal can be aligned towards the desired location.
Because of their potential to transform a hostile wireless environment into an adaptive and favorable propagating channel, RISs have received much attention from the research community. RIS has the potential to enhance spectral efficiency (SE) significantly, and energy efficiency (EE) due to the large number of passive REs \cite{8319526}.

On the other hand, device-to-device (D2D) communication is also considered a promising technology proposed in the 5G standard that enables direct communications between D2D users. In D2D communication, the same time-frequency resources of cellular users are re-utilized by the D2D users, thus, allowing massive access without aggravating the spectrum crunch \cite{6805125}. 
However, in D2D communication, a successful transmission is highly reliant on the propagation environment due to the limited power budget available at nodes. Unfortunately, this restriction limits the applicability of D2D in many of the existing scenarios, especially in dense urban environments. Since RIS can adapt an unknown channel to a favorable propagation environment, deploying RISs can effectively alleviate this constraint \cite{9549117}. 
In \cite{9519722}, the authors have optimized the RIS-aided underlay D2D communication to maximize the capacity by optimizing RIS phase shifts along with spectrum reuse and transmit beamforming. A joint resource allocation to maximize the sum rate of a RIS-assisted D2D underlay cellular network was studied in \cite{9507508}. Likewise, in \cite{9424402}, performance analysis for RIS-assisted D2D communication was carried out for underlay and overlay modes.

This paper investigates the performance of a RIS-empowered NOMA-based D2D communication system. The proposed scenario considers a downlink network, where a user nearby the base station (BS), is utilized to serve as a D2D transmitter (DT), facilitating the communication with a pair of users, i.e., D2D receivers (DRs), which were otherwise not accessible by BS. DT is deployed with RIS, which comprises $M$ REs. To support both DRs, a hard partitioning-based approach is utilized at RIS. 
Unlike \cite{9424402}, where the point-to-point D2D communication without a direct link was considered, we consider the novel RIS-empowered NOMA-based D2D communication with both direct and reflected links.  
This work's main contribution can be summarized as follows: 
\begin{itemize}
\item We obtain the closed-form expressions for the upper and lower bounds of ergodic rate for the NOMA pair of the proposed RIS-enabled NOMA-based D2D communication system. Initially, we formulate the received signal-to-interference-plus-noise-ratio (SINR) and then utilize it for deriving the closed-form expressions of SE and EE for both the DRs;
\item In addition, we illustrate the effect of the distribution of REs, the power allocation ratio, and the Nakagami-$m$ fading parameters on network performance;
\item Lastly, the proposed RIS-enabled NOMA-based D2D network is compared to its corresponding OMA counterpart and the case without RIS. 
\end{itemize}

\begin{figure}[!t]
\centering
\includegraphics[width=0.7\columnwidth]{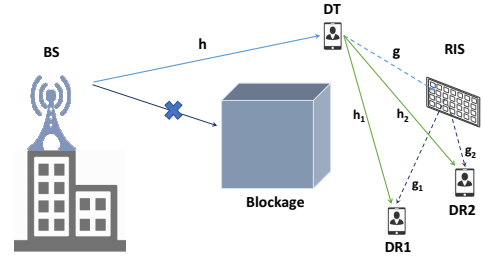} 
\caption{Schematic for RIS-empowered D2D Communication.}
\label{fig:Blkdig}
\vspace{-0.5cm}
\end{figure}

\section{System Model}
Fig. \ref{fig:Blkdig} illustrates the system model where a single antenna BS tries to communicate with a pair of blocked UEs, i.e., D2D receivers (DRs), denoted as DR1 and DR2. So, a D2D transmitter (DT) is utilized to set a reliable communication link. Further, the DT-to-DR transmission is assisted via RIS with $M$ number of REs. The system model can be regarded as a connected D2D-enabled cellular system, where a cluster of UEs are brought into coverage through the nearest connected UE\footnote{There can be multiple users within the cluster, however, due to complexity requirements, we restrict ourselves to the two-user case, i.e., two DRs \cite{9345507}.}. Further, without losing any generality, in this work, we focus mainly on D2D communication, i.e., communication from DT to DRs.

In order to support two DRs, RIS is partitioned in two sub-RISs, each having $M_1$ and $M_2$ number of REs, with $M_1=\eta\,M$, $M_2=\left(1-\eta\right)\,M$, $M_1+M_2=M$ and $\eta$ being the allocation parameter. 
Similar to \cite{9240028, 9000593}, a quasi-static and flat fading channel is assumed with known channel state information (CSI). Further, the BS-to-DT, DT-to-RIS and RIS-to-DR channel links can either be line-of-sight (LoS), or non-LoS (NLoS) and thus characterized through Nakagami-$m$ fading model \cite{9345507}. The elements of $\textbf{g}$, $\textbf{g}_1$ and $\textbf{g}_2$ follow the Nakagami-$m$ fading model with $m_0$, $m_1$ and $m_2$ as the fading parameters.  
Similarly, the direct link between DT-to-DR is also characterized through Nakagami-$m$ fading channel with $m_{h_l}$ as fading parameter, where $l = {1, 2}$.

In accordance with the NOMA and RIS concepts, the received signals at DR1, $r_1$, and DR2, $r_2$, can be expressed as
\begin{align}
    r_1 &= \left(h_1 + \mathbf{\bar{g}}_1\:\mathbf{\Phi}_1 \:\mathbf{g}_1\,\right)\,(\beta _1 \,x_1 + \beta _2 \,x_2)\,\sqrt{P_r} + N_o, \\
    r_2 &= \left(h_2 + \mathbf{\bar{g}}_2\:\mathbf{\Phi}_2 \:\mathbf{g}_2\,\right)\, (\beta_1\, x_1 + \beta_2\, x_2)\,\sqrt{P_r} + N_o,
\end{align}
where $x_1$ and $x_2$ represent the desired DR1 and DR2 signals, respectively. Likewise, $\beta_1$ and $\beta_2$ is the power allocation coefficient for DR1 and DR2. Further, $\beta_1$ and $\beta_2$ adhere to the NOMA constraint, i.e., $\beta^2_1\; + \;\beta^2_2 = 1$. Further, $P_r$ and $N_o$ denote the transmit power at DT and the additive white Gaussian noise (AWGN), with $N_o \in \mathcal{CN}(0,\sigma^2)$.

Now the received signal of the DRs can be maximized through proper phase shifting at the RIS. Mathematically, $ \left |{\mathbf{g} \mathbf{\Phi}_1 \mathbf{g}_1}\right | = \left | \sum _{i=1}^{M_1} { g^i \delta^i e^{j\theta_i}\,g_1^{i}}\right|$ and $ \left |{\mathbf{g} \mathbf{\Phi}_2 \mathbf{g}_2}\right | = \left | \sum _{i=1}^{M_2} { g^i \delta^i\,e^{j\theta_i} g_2^{i} }\right|$ maximizes the received signal power at DRs. Here, $g^i$, $g_1^i$ and $g_2^i$ denotes the $i$-th element of $\mathbf{g}$, $\mathbf{g}_1$ and $\mathbf{g}_2$, respectively. Thus, re-configuring $\theta^i$ to $\tilde{\theta}$ maximize the received power. The corresponding channel gain\footnote{Without losing any generality, $\delta^i = 1, \forall\, i$ is assumed.} to $\tilde{\theta}$ can be expressed as
{\allowdisplaybreaks
\begin{align}
\left |\mathcal{H}_1\right|^2 & = \left |h_1 + {\mathbf{\bar{g}}_1 \mathbf{\Phi}_1 \mathbf{g}_1} \right|^2 = \left(\left|h_1 \right| + \sum _{i=1}^{M_1} \left|\bar{g}_1^i\right| \left|g_1^i\right| \right)^2,
\label{ch_G1} 
 \\
\left |\mathcal{H}_2\right|^2 & = \left |h_2 + {\mathbf{\bar{g}}_2 \mathbf{\Phi}_2 \mathbf{g}_2} \right|^2 = \left(\left|h_2\right| + \sum _{i=1}^{M_2} \left|\bar{g}_2^i\right| \left|g_2^i\right| \right)^2.
\label{ch_G2}
\end{align}}

\section{Performance Analysis}
This section evaluates the bound on the ergodic rate of DRs. Further, the SE and EE for RIS-enabled NOMA-based D2D is formulated considering the fading parameter, power allocation, and REs distribution. 
Initially, the SINR for both the DRs is formulated and later on we utilize it in evaluating the SE and EE.

Considering the DR1 signal as an interference, DR2 will decode the received signal with the following SINR
\begin{align}
   SINR_{DR2} & =  \frac{\left |\mathcal{H}_2\right|^2 \beta _2^2 \,P_r} {\left |\mathcal{H}_2\right|^2 \beta _1^2 \,P_r + N_o} = \frac{\left |\mathcal{H}_2\right|^2 \beta _2^2 \,\rho_r} {\left |\mathcal{H}_2\right|^2 \beta _1^2 \,\rho_r + 1},
   \label{snrr_dr2}
\end{align}
where $\rho_r=P_r/N_o$ is transmit SNR at DT.

Likewise, at DR1, applying SIC, initially, DR1 will decode the received signal of DR2. SINR for it can be expressed as
\begin{equation}
   SINR_{DR1 \to DR2} =  \frac{\left |\mathcal{H}_1\right|^2 \beta _2^2 \,\rho_r} {\left |\mathcal{H}_1\right|^2 \beta _1^2\, \rho_r + \sigma^2}.
   \label{snr_dr1}
\end{equation}
After decoding and canceling the signal of DR2, DR1 can decode its own signal with SINR of
\begin{equation}
   SINR_{DR1} = {\beta _1^2\,\left |\mathcal{H}_1\right|^2  \,\rho_r}.
   \label{snr_dr11}
\end{equation}

\subsection{Channel Characterization}
Now the channel gains, $\mathcal{H}_1$ and $\mathcal{H}_2$, as defined in \eqref{ch_G1} and \eqref{ch_G2}, respectively, do not conform to any known closed-form distribution. Thus, for the sake of simplification of the analytical performance, we can approximate $G_1$ and $G_2$ (for $M_1 \gg 1$ and $M_2 \gg 1$) as $\left |\mathcal{H}_1\right|^2 = \left(\sum _{i=1}^{M_1} \left|\bar{g}_1^i\right| \left|g_1^i\right| \right)^2$ and $\left |\mathcal{H}_2\right|^2 = \left(\sum _{i=1}^{M_2} \left|\bar{g}_2^i\right| \left|g_2^i\right| \right)^2$, respectively.   
The distribution function for the channel gain, $\left |\mathcal{H}\right|^2$, can be defined for $ g\geq 0$ as \cite{samuh2020performance} 
\begin{align} 
f_{\left |\mathcal{H}\right|^2}(y) & = \frac {\sqrt{y^{a}}}{2b^{a+1}\Gamma (a+1)} \exp \left ({-\frac {\sqrt{y}}{b}}\right), 
\label{pdf1} \\
F_{\left |\mathcal{H}\right|^2}(y) & = \frac {\gamma \left ({a+1,\frac {\sqrt{y}}{b}}\right)}{\Gamma \left ({a+1}\right)\sqrt{y}}.
\label{cdf1}
\end{align}
Here, $a$ and $b$ are the variables defined as $a=\frac{m_0 \Gamma(m_0)^2 M m_{l} \Gamma(m_{l})^2}{m_0 \Gamma(m_0)^2 m_{l}\Gamma(m_{l})^2 - \Gamma(m_0+\frac{1}{2})^2 \Gamma(m_{l}+\frac{1}{2})^2}-N-1$ and $b=\frac{m_0 \Gamma(m_0)^2 m_{l} \Gamma(m_{l})^2 - \Gamma(m_0+\frac{1}{2})^2 \Gamma(m_{l}+\frac{1}{2})^2}{\sqrt{\frac{m_0}{\Omega_{g}}} \Gamma(m_0)\Gamma(m_0+\frac{1}{2}) \sqrt{\frac{m_{l}}{\Omega_{g_l}}}\Gamma(m_{l})\Gamma(m_{l}+\frac{1}{2})}$, with $N \in \{M_1, M_2\}$ and $l\in\{1, 2\}$, for $\mathcal{H} \in \{\mathcal{H}_1, \mathcal{H}_2\}$. Further, $\Gamma(\cdot)$ represents the Gamma function and $\gamma(\cdot,\cdot)$ indicates the lower incomplete Gamma function.

\begin{figure*}[!t]
\begin{align} 
R^{l}_{DR1} =&\,\frac {1}{\ln{(2)}\,\Gamma (a+1)}  \left[  \frac {\pi\, \mathrm {csc} \left({a\pi}/{2}\right)\, \mathcal{F}_1}{\left(a+2\right) \left(\beta_1 b\right)^{a+2}\,\left(\rho_{r}\right)^ {\frac {a}{2}+1}}      +\frac {\pi\,  \mathrm {sec}\left ({{a\pi}/{2}}\right)\,\mathcal{F}_1 } {\left(a+1\right)\, \left(\beta_1 b\right)^{a+1} \left(\rho_{r}\right)^ {\frac {a+1}{2}}} \right. \nonumber\\
&\left.\qquad + 2\,a\,\left(a-1\right) \,\psi ^{(0)} \left(a+1\right) + \,\left(a^2-a\right)\,\Gamma(a-1)\,\ln\left( {b^{2}\, \beta_1^2\,\rho_{r}}\right) + \Gamma (a-1) \, \mathcal{F}_3  \right] \label{e14}
\end{align} \vspace{-0.3cm}
\hrule
\vspace{-0.3cm}
\end{figure*}
\begin{figure*}[!t]
\begin{align} 
R^{l}_{DR2} = &\,\frac {1}{\ln{(2)}\Gamma (a+1)}  \left[ \frac {\pi \, \mathrm {csc} \left( {a\pi}/{2}\right)}{\left(a+2\right) b^{a+2}}  \left \{ \frac{\mathcal{F}_4}{(c_1 \rho_{r})^{\frac {a}{2}+1}} -\frac{\mathcal{F}_5}{(c_2 \rho_{r})^{\frac {a}{2}+1}} \, \right \} + \frac {\pi \,  \mathrm {sec}\left ({{a\pi }/{2}}\right)} {\left(a+1\right) b^{a+1}}  \left \{ \frac{\mathcal{F}_6 } {(c_1\rho_{r})^{\frac {a+1}{2}}} - \frac{\mathcal{F}_7}{(c_2\rho_{r})^{\frac {a+1}{2}}}  \right \} \right.\nonumber\\
&\left. \qquad \qquad \qquad \qquad \qquad \qquad \qquad \qquad +\,  \Gamma (a-1) \left \{  \mathcal{F}_8 - \mathcal{F}_9 \right \} + \left(a^2-a\right)\Gamma(a-1)\ln\left( \frac{c_1}{c_2} \right)   \right] \label{e15}
\end{align} 
\vspace{-0.2cm}
\hrule
\vspace{-0.3cm}
\end{figure*}
\begin{figure*}[!t]
\begin{align}
R^u_{DR1} = & \log_2 \left[ 1 + \Xi_1\,{\Omega_{h_1}} + M_1\,\Xi_1\, {\Omega_{m_0}\,\Omega_{m_1} }  + M_1\left(M_1-1\right)\,\Xi_1\, \frac{\Omega_{m_0}}{m_0} \, \left\{ \frac {\Gamma \left(m_0 +\frac{1}{2} \right)}{\Gamma \left(m_0 \right)} \right\}^2 \,\frac{\Omega_{m_1}} {m_1} \, \left\{ \frac{ \Gamma \left(m_1 +\frac{1}{2} \right) } { \Gamma \left(m_1 \right) } \right\}^2  \right. \nonumber\\
& \left. \qquad \qquad\qquad \qquad\qquad  +\, 2\, M_{1}\, \Xi_1\, \frac{\Gamma(m_{h_1}+\frac{1}{2})}{\Gamma(m_{h_1})} \, \sqrt{\frac{\Omega_{m_{h_1}}}{m_{h_1}}}\,  \frac{\Gamma(m_0+\frac{1}{2})}{\Gamma(m_0)}\, \sqrt{\frac{\Omega_{m_0}}{ m_0}} \, \frac{\Gamma(m_1+\frac{1}{2})}{\Gamma(m_1)}\,  \sqrt{\frac{\Omega_{m_1}}{m_1}} \right] \label{e16} 
\end{align} 
\vspace{-0.2cm}
\hrule
\vspace{-0.3cm}
\end{figure*}
\begin{figure*}[!t]
\begin{align}
R^u_{DR2}  & = \log_2 \left[ 1 + \Xi_2 {\Omega_{h_2}} + M_2 \Xi_2  {\Omega_{m_0} \Omega_{m_2} }  + \frac {M_2\left(M_2-1\right) \Xi_2    \Omega_{m_0} \left\{  \Gamma \left(m_0 +\frac{1}{2} \right) \right\}^2 \Omega_{m_2} \left\{  \Gamma \left( m_2 + \frac{1}{2} \right) \right\}^2 }{ m_0  \left \{  \Gamma \left(m_0\right) \right\}^2 m_2  \left\{  \Gamma \left(m_2 \right) \right\}^2 } \right. \nonumber\\
&~~~ \left.  + \frac { 2  M_{2}  \Xi_2   \Gamma(m_{h_2}+\frac{1}{2})  \Gamma(m_0+\frac{1}{2})   \Gamma(m_2+\frac{1}{2})  \sqrt{\Omega_{m_{h_2}}\Omega_{m_0} \Omega_{m_2}}  } {  \Gamma(m_{h_2})   \Gamma(m_0)  \Gamma(m_2)  \sqrt{m_{h_2}m_0m_2}} \right] - \log_2 \left[ 1 + \Xi_3 {\Omega_{h_2}} + M_2 \Xi_3  {\Omega_{m_0} \Omega_{m_2} } \right. \nonumber\\
& \left. \hspace{-1cm}  +  \frac {M_2\left(M_2\!-\!1\right) \Xi_3   \Omega_{m_0}\! \left\{  \Gamma\! \left(m_0 \!+\!\frac{1}{2} \right)\! \right\}^2 \Omega_{m_2}\! \left\{  \Gamma\! \left( m_2\! +\! \frac{1}{2} \right)\! \right\}^2 }{ m_0  \left \{  \Gamma \left(m_0\right) \right\}^2 m_2  \left\{  \Gamma \left(m_2 \right) \right\}^2 } \! +\!  \frac {2  M_{2}  \Xi_3  \Gamma(m_{h_2}\!+\!\frac{1}{2})  \Gamma(m_0\!+\!\frac{1}{2})   \Gamma(m_2\!+\!\frac{1}{2})  \sqrt{\Omega_{m_{h_2}}\!\Omega_{m_0}\! \Omega_{m_2}}} { \Gamma(m_{h_2})   \Gamma(m_0)  \Gamma(m_2)   \sqrt{m_{h_2}m_0m_2}} \right] \label{e17} 
\end{align} 
\vspace{-0.15cm}
\hrule
\vspace{-0.3cm}
\end{figure*}

\subsection{Ergodic Rate}
The ergodic rates for DR1 and DR2 can be formulated as
    $R_{DR1} = \mathbb{E}\left[\log_2 \left( 1+ SINR_{DR1}\right)\right]$ and
    $R_{DR2} = \mathbb{E}\left[\log_2 \left( 1+ SINR_{DR2}\right)\right]$, respectively.
Since the channel gain's exact distribution is unknown, the expectations are mathematically intractable, and thus a closed-form expression may not be derived. Hence, we resort to approximating the ergodic rates of DRs with tight upper and lower bounds. Specifically, the upper bound is derived by invoking Jensen's inequality, and the lower bound is derived by utilizing the approximate PDF as described in \eqref{pdf1}. 
The bounds on the ergodic rate of DRs are evaluated following a series of mathematical manipulations. The following Lemmas present the upper and lower bound for the proposed RIS-enabled NOMA-based D2D communication system.

\begin{lemma}
The lower bound on the ergodic rates of DR1 and DR2 can be expressed as in \eqref{e14} and \eqref{e15}, shown on the top of the next page.
\end{lemma}
\begin{proof}
The proof is presented in Appendix A.
\end{proof}

\begin{lemma}
Likewise, the upper bound on the ergodic rate of DR1 and DR2 can be expressed as in \eqref{e16} and \eqref{e17}, shown on the top of the page.
\end{lemma}
\begin{proof}
The proof is presented in Appendix B.
\end{proof}

\subsection{SE and EE}
Based on the ergodic rate established in the preceding subsection, SE of RIS-enabled NOMA-based D2D can be described as $SE = R_{DR1} + R_{DR2}$. 
Similarly, the EE can be defined as the ratio of the SE to the total power utilized, $P_{\rm tot}$, in bits/Joule/Hz. $P_{\rm tot}$ consists of the power utilized by the BS, DT, RIS, and DRs. Thus, the EE may be represented as 
\begin{align}
EE =\, \frac{SE}{P_{\rm tot}}=\, \frac{SE}{(1+ \alpha) P_r + M P_{RE} +  2P_U},
\label{ee2}
\end{align}
where $P_{r}$ denotes the static power consumption of DT. Likewise, $\alpha P_r$ is the dynamic power consumption at DT. Further, $P_{RE}$ denotes the power consumed by each of the RE and $P_U$ is the power utilized by DR.

\begin{figure}[!t]
\centering
\includegraphics[width=7cm, height=5cm]{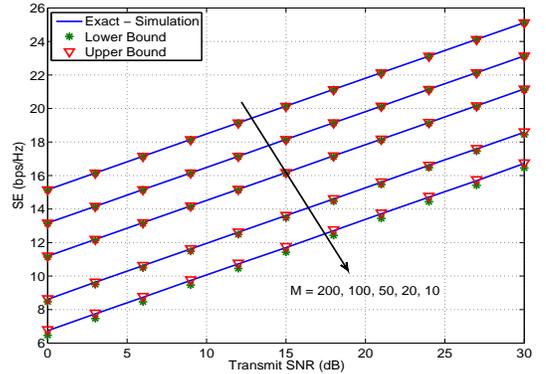} 
\caption{SE versus transmit SNR with varying $M$.}
\label{fig:SE1}
\vspace{-0.2cm}
\end{figure}

\section{Simulation Result}
This section presents the simulation and analytical results for the proposed RIS-empowered NOMA-based D2D communication system. For the direct links, the Nakagami fading parameters are assumed to be $m_{h_1} = m_{h_2} = 2$ and, for the RIS reflected links, $m_0 = m_1 = m_2 = 5$, respectively. Likewise, the power allocation factor for DR1 is $\beta ^2_1 = 0.3$ while for DR2 is $\beta ^2_2 = 0.7$, if not specified otherwise. Additionally, the value of RE allocation parameter $\eta$ is assumed to be $0.5$.

Fig. \ref{fig:SE1} shows the SE results for the proposed RIS-empowered NOMA-based D2D communication. Specifically, it shows SE with respect to the transmit power while comparing the simulation and analytical results. These results can easily infer the following observations: 1) Apart from smaller $M$, analytical SE is quite precise compared to simulation-based SE. 
2) Due to the multiplicative path-loss, for less number of REs, i.e., smaller $M$, the received power from the direct link is significant. However, as the number of REs increases, the received power from a RIS-reflected link is much more than the power received from the direct link to the extent that it can safely be ignored. Thus, it can easily be inferred from the analytical and simulation framework that the received signal power from the direct link is relatively insignificant and can be ignored as compared to the received power from the RIS-reflected link. 

\begin{figure}[!t]
\centering
\includegraphics[width=7 cm, height = 5 cm]{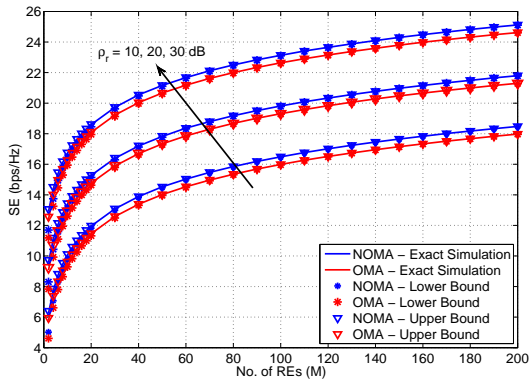}
\caption{SE of NOMA-/OMA-based D2D with respect to $M$.}
\label{fig:SE2}
\vspace{-0.3cm}
\end{figure}
\begin{figure}[!t]
	\centering
	\includegraphics[width=7 cm, height = 5 cm]{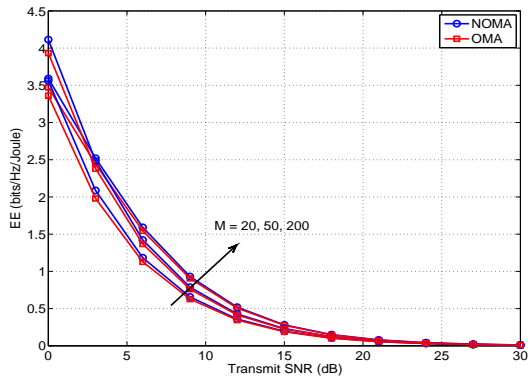}
	\caption{EE versus transmit SNR at different $M$.}
	\label{fig:EE1}
	\vspace{-0.5cm}
\end{figure}

Fig. \ref{fig:SE2} shows the SE of the proposed RIS-empowered D2D communication system for both NOMA and OMA scenarios. Specifically, the SE of both NOMA and OMA scenarios is plotted with respect to the number of REs for different SNRs. It can easily be observed here that, instead of increasing the transmit power, the number of REs can be increased to get the same SE. 
In other words, for fixed required SE, we can tradeoff the transmit power with the number of REs. As D2D users are usually power constrained, RIS-empowered D2D communication can be a viable alternative to cut down the transmit power and improve network EE. 
Further, as evident from the result, RIS-empowered NOMA-based D2D is more spectrally efficient as compared to OMA-based D2D. 
For instance, SE at $20$ dB SNR and $M=20$ is $15.26$ bps/Hz for NOMA and 14.76 bps/Hz for OMA, respectively. 
The NOMA gain will increase when the channel gain between UEs increases. Further, SE also improves with the number of REs, as evident from the result. Likewise, SE grows as the transmit power goes up.

Fig. \ref{fig:EE1} shows EE of the proposed RIS-empowered D2D communication system, where EE of both NOMA and OMA scenarios is plotted with respect to the transmit power for the varying number of REs. Further, it can be inferred that RIS-empowered D2D is energy-efficient as compared to OMA. Further, it can also be observed that the EE increases with the number of REs, whereas EE decreases as the transmit power increases. 
This is because SE increases linearly while the transmit power increases logarithmically; thus, the overall compounding impact decreases EE while increasing transmit SNR. 
In addition, EE is likewise saturated for a large number of REs, and no further gains are observed.
The result also demonstrates that increasing the number of REs does not improve performance, as SE increases while EE becomes saturated.
Thus, it can be inferred that RIS improves the SE and EE performance of the D2D system.



\section{Conclusion}

In this paper, we investigated the performance of a RIS-empowered NOMA-based D2D communication system. Specifically, we derived the closed-form expressions for SE's upper and lower bounds. As shown through the results, apart from the smaller values of the number of REs, the bounds are pretty tight and converge to exact SE, even for moderate REs. Further, we have also investigated the EE performance. Since the D2D devices are usually power-constrained, the results show that the transmit power can be a tradeoff with the number of REs at RIS. Additionally, the results are also compared with the OMA scenario, where it has been shown that NOMA-based D2D outperforms the OMA-based case.  

\section{Acknowledgement}
This work was supported by the Nazarbayev University CRP Grant no. 11022021CRP1513.

\appendices
\section{Proof of Lemma 1}
The ergodic rate of DR1 can be formulated as 
\begin{equation}
R_{DR1} =\frac {1}{\ln\left(2\right)} \int \limits _{0}^\infty \underset{\mathcal{J}_1}{\underbrace{{\ln \left(1+ \beta_1^2\left |y \right |^2 \rho_r\right)} f_{\left |H\right|^2} (y) \mathrm{d}y}}.
\label{ap_R1_2}
\end{equation}
Further, using \eqref{pdf1} and the below relation 
$\ln (t) = (t-1) \,_{2} F_{1} \left ({1, 1;2; 1-t}\right)$,
$\mathcal{J}_1$ in \eqref{ap_R1_2} can be modified as 
\begin{align}
\mathcal{J}_1 = \frac{1}{\Gamma(a+1)} \int \limits _{0}^\infty \frac{\sqrt{g}\,_{2}F_{1}\left ({1, 1;2; -g}\right) }{2 b^{a+1}}
e^{\left(-\frac{\sqrt{g}}{b}\right)}
\mathrm{d} g.
\end{align}
Here, $\,_{2}F_{1}(\cdot,\cdot;\cdot;\cdot)$ represents the Gauss hyper-geometric function. Now, this $\mathcal{J}_1$ can be solved utilizing \cite[Theorem 3]{samuh2020performance} and substituted in \eqref{ap_R1_2}. After rearranging the terms, the ergodic rate for DR1 can be given as shown in \eqref{e15}, where $\mathcal{F}_1\! =\! {}_{1}F_{2}\left( {1\!+\!\frac {a}{2}; \frac {3}{2}, 2\!+\!\frac {a}{2}; \frac {-1}{4b^{2} \beta_1^2\rho_{r}} }\right)$, 
$\mathcal{F}_2\! =\! {}_{1}F_{2}\left ({\frac {a\!+\!1}{2};\frac {1}{2},\frac {a\!+\!3}{2};  \frac {-1}{4  b^{2} \beta_1^2  \rho_{r}}}\right)$ and
$\mathcal{F}_3 \!=\!  {}_{2}F_{3}\left ({1,1; 2, 1\!-\!\frac {a}{2}, \frac {3-a}{2};   \frac {-1}{4  b^{2}  \beta_1^2  \rho_{r}}}\right)$

Likewise, the ergodic rate for DR2, $R_{DR2}$, can be given by  
\begin{align}
{R_{DR2}} =& \frac {1}{\ln \left ({2 }\right)} \left\{\int \limits _{0}^\infty {\ln{\left(1+c_1 \left |y \right |^2 \rho_r\right)}} f_{\left |\mathcal{H}\right|^2}(y)\mathrm {d}y \nonumber \right. 
\\ & \qquad  -  \left.  \int \limits _{0}^\infty {\ln{\left(1+c_2\left |y \right |^2 \rho_r\right)}} f_{\left |\mathcal{H}\right|^2}(y)\mathrm {d}y \right\},
\end{align}
where $c_1$ and $c_2$ are defined as $c_1 = \beta_1^2 + \beta_2^2$ and $c_2 = \beta _1^2$.
It can be evaluated similarly to $\mathcal{J}_1$. After rearranging the terms, the ergodic rate for DR2 can be given as shown in \eqref{e16}, where $\mathcal{F}_4 \!=\! {}_{1}F_{2}\left( {1\!+\!\frac {a}{2}; \frac {3}{2}, 2\!+\!\frac {a}{2};  \frac {-1}{4 b^{2}  c_1 \rho_{r}} }\right)$,
$\mathcal{F}_5 \!=\!{}_{1}F_{2}\left ({1\!+\!\frac {a}{2}; \frac {3}{2}, 2\!+\! \frac {a}{2};  \frac {-1}{4 b^{2}  c_2  \rho_{r}} }\right)$,
$\mathcal{F}_6\! =\!{}_{1}F_{2}\left ({\frac {a+1}{2};\frac {1}{2},\frac {a+3}{2};  \frac {-1}{4  b^{2}  c_1  \rho_{r}}}\right)$,
$\mathcal{F}_7\! =\!{}_{1}F_{2}\left ({\frac {a+1}{2}; \frac {1}{2}, \frac {a+3}{2};  \frac {-1}{4  b^{2}  c_2  \rho_{r}}} \right)$,
$\mathcal{F}_8\! =\! {}_{2}F_{3}\left ({1, 1;  2,  1\!-\!\frac {a}{2},  \frac {3-a}{2};   \frac {-1}{4  b^{2}  c_1  \rho_{r}}}\right)$ and
$\mathcal{F}_9 \!=\! {}_{2}F_{3}\left ({1,  1;  2,  1\!-\!\frac {a}{2},  \frac {3-a}{2};    \frac {-1}{4  b^{2}  c_2  \rho_{r}}}\right)$.
This completes the proof of Lemma 1.

\section{Proof of Lemma 2}
Applying Jensen's inequality, we define the upper bound for DR1 as ${R}^{u}_{DR1}$, where 
$R_{DR1} \leq {R}^{u}_{DR1}$, with $\Xi_1 = \beta_1^2 \rho_r$, as
\begin{align}
{R}^{u}_{DR1} =  \log_2 \left( 1 + \Xi_1 \mathbb{E} \left[\left |\mathcal{H}_1\right |^2\right] \right).
\label{eqn_R_ub_2}
\end{align}
To calculate $\mathbb{E}\left[  \left |\mathcal{H}_1\right |^2 \right]$, we apply the binomial expansion theorem (BET) as 
\begin{align}
&\mathbb{E}\left[ \left |\mathcal{H}_1\right |^2 \right] = 
\mathbb{E}\left[ \left |  h_1 + \,\sum _{i=1}^{M_1} \left|\bar{g}_1^i\right| \left|g_1^i\right| \right |^2 \right] 
= \underbrace { \mathbb {E} \left \{{\left| h_1 \right|^{2} }\right \} }_{\mathcal{E}_{1}} 
\nonumber \\ 
&   + \, \underbrace { \mathbb {E} \left \{{ \left (\sum _{i=1}^{M_1} \left|\bar{g}_1^i\right| \left|g_1^i\right| \right)^{2} }\right \} }_{\mathcal{E}_{2}} + 2 \, \underbrace {\mathbb {E}\left\{{ \sum _{i=1}^{M_1} \left|\bar{g}_1^i\right| \left|g_1^i\right|\left|h_1\right|}\right \} }_{\mathcal{E}_{3}}. 
\end{align}

Now, we have $\mathcal{E}_{1} = \Omega_{h_1}$. Likewise, to calculate $\mathcal{E}_{2}$, we apply BET again; thus, on expanding, $\mathcal{E}_{2}$ can be expressed as
\begin{align}
\mathbb {E}\left \{{ \sum _{i=1}^{M_1} \left|\bar{g}_1^i\right|^2 \left|g_1^i\right|^{2}}\right \}  +\, \mathbb {E}\left \{{\sum _{i=1}^{M_1} \sum _{\substack {i=1 \\ j\ne i}}^{M_1} \left|\bar{g}_1^i\right| \left|g_1^i\right|\left|\bar{g}_1^j\right| \left|g_1^j\right|}\right \}, \end{align}
where $\mathbb {E}\left \{ \sum _{i=1}^{M_1} \left|\bar{g}_1^i\right|^{2} \left|{g}_{1}^i\right|^{2}\right \} = {M}_1\Omega_{m_0}\Omega_{m_1}$. 
Further, for $\mathbb {E}\left \{{\sum _{i=1}^{M_1} \sum _{\substack {i=1 \\ j\ne i}}^{M_1} \left|\bar{g}_1^i\right| \left|g_1^i\right|\left|\bar{g}_1^j\right| \left|g_1^j\right|}\right \}$, the expected value of a Nakagami-$m$ variable can be given as
$\mathbb {E}\{\left|g_{1}\right|\} =\frac {\Gamma(m_1+\frac{1}{2})}{\Gamma(m_1)} \sqrt{\left(\frac{\Omega_{m_1}}{m_1}\right)}$.
Since $\bar{g}_1^i$ and $g_1^i$ are mutually independent, we can have 
{\allowdisplaybreaks
\begin{align}
\mathbb {E}&\left\{{\sum _{i=1}^{M_1} \sum_{\substack {j=1 \\ j\ne i}}^{M_1} \left|{\bar{g}}_{1}^i\right| \left|{g}_{1}^i\right| \left|\bar{g}_{1}^j\right| \left| g_{2}^j\right|}\right\} = M_1\left(M_1-1\right) \left\{\frac{\Omega_{m_0}}{m_0}\right\} 
\nonumber\\
&\times \frac{ \left \{\Gamma\left(m_0+\frac{1}{2}\right)\right\}^2 \left \{ \Gamma \left( m_1+\frac{1}{2} \right)\right\}^2 }{\left \{\Gamma\left(m_0\right)\right\}^2 \left \{\Gamma\left(m_1\right)\right\}^2} 
\left\{\frac{\Omega_{m_1}}{m_1}\right\}. 
\end{align}}
Likewise, $\mathcal{E}_{3}$ can be calculated as
\begin{align} 
&\mathcal{E}_{3} = M\, \Gamma\left(m_{h_1}+\frac{1}{2}\right)\Gamma\left(m_0+\frac{1}{2}\right) \Gamma\left(m_1+\frac{1}{2}\right)
\nonumber \\ 
&\times \sqrt{\frac{\Omega_{m_{h_1}} \Omega_{m_0} \Omega_{m_1}} {m_{h_1} m_0 m_1}} / 
\left[
\Gamma\left(m_{h_1}\right) \Gamma\left(m_0\right) \Gamma\left(m_1\right)
\right]. 
\end{align}

Finally, putting $\mathcal{E}_{1}$, $\mathcal{E}_{2}$ and $\mathcal{E}_{3}$
all together yields $\mathbb{E}\left[  \left |\mathcal{H}_1\right |^2 \right]$ 
which can be put in \eqref{eqn_R_ub_2} to give the desired upper bound as shown in (27).

Likewise, the upper bound on the ergodic rate of DR2, ${R}^{u}_{DR2}$ can be defined as 
$R_{DR2} \leq {R}^{u}_{DR2},$    
where $R_{DR2}$ can be defined as 
\begin{align}
{R}_{DR2} &= \mathbb{E}\left[\log_2 \left( \frac{1+ \left |\mathcal{H}_2\right |^2 \left(\beta_1^2 + \beta_2^2 \right) \,\rho_r} {1 + \left |\mathcal{H}_2 \right |^2 \beta_{1}^2 \rho_r}\right)\right], \nonumber \\
&= \mathbb{E}\left[ \log_2 \left( 1 + \Xi_2 \left |\mathcal{H}_2\right |^2\right) - \log_2 \left( 1 + \Xi_3 \left |\mathcal{H}_2\right |^2 \right) \right], 
\label{eqn_Rn_4}
\end{align}
where $\Xi_2 = \left(\beta_1^2 +\beta_2^2 \right) \rho_r$ and $\Xi_3 = \beta_1^2 \rho_r$. Thus, ${R}^{u}_{DR2}$ can be defined as 
\begin{align}
{R}^{u}_{DR2} \hspace{-0.1cm} = \hspace{-0.1cm} \log_2 \left( 1 \hspace{-0.05cm} + \hspace{-0.05cm} \Xi_2 \mathbb{E} \left[\left |\mathcal{H}_2\right |^2\right] \right) \hspace{-0.05cm} - \hspace{-0.05cm} \log_2 \left( 1 \hspace{-0.05cm} + \hspace{-0.05cm} \Xi_3 \mathbb{E} \left[\left |\mathcal{H}_2\right |^2\right] \right).
\label{eqn_R_ub_3}
\end{align}
Similar to $\mathbb{E} \left[\left |\mathcal{H}_1\right |^2\right]$, $\mathbb{E} \left[\left |\mathcal{H}_2\right |^2\right]$ can be evaluated. After substituting and rearranging the terms, ${R}^{u}_{DR2}$ is given in \eqref{e17}.

\bibliographystyle{IEEEtran}
\bibliography{Bibil}
\end{document}